\documentclass[a4paper, cleveref, autoref, USenglish, thm-restate,pdfa]{lipics-v2021}

\usepackage{listings}
\usepackage[color=orange!25]{todonotes}
\usepackage{amsthm}
\usepackage{amssymb, mathtools}
\usepackage{amsmath}
\allowdisplaybreaks
\usepackage{mleftright}
\usepackage{tikz}
\usetikzlibrary{calc,intersections,through,arrows.meta,external}
\usepackage{tkz-euclide}
\usepackage{booktabs}
\usepackage{multirow}
\usepackage{xspace}
\usepackage{stmaryrd}
\usepackage{datetime}
\usepackage{graphicx}
\usepackage{microtype}
\usepackage{spverbatim}
\usepackage{rotating}
\usepackage{cancel}
\usepackage{empheq}
\usepackage{nicefrac}
\usepackage[title]{appendix}
\usepackage{aligned-overset}
\usepackage{csquotes}

\newcommand{\dir}{\mathrm{dir}}
\newcommand{\none}{\mathrm{none}}

\newcommand{\cl}{\textrm{cl}}

\pdfoutput=1 
\hideLIPIcs  

\title{Label Correcting Algorithms for the Multiobjective Temporal Shortest Path Problem} 

\author{Edina {Marica}}{Technical University of Munich, Campus Straubing for Biotechnology and Sustainability, Professorship of Optimization and Sustainable Decision Making, Am Essigberg 3, 94315 Straubing, Germany \and \url{https://odm.cs.tum.de/en/} }{edina.marica@tum.de}{https://orcid.org/0009-0003-0594-000X}{}

\author{Clemens {Thielen}}{Technical University of Munich, Campus Straubing for Biotechnology and Sustainability, Professorship of Optimization and Sustainable Decision Making, Am Essigberg 3, 94315 Straubing, Germany \and \url{https://odm.cs.tum.de/en/} }{clemens.thielen@tum.de}{https://orcid.org/0000-0003-0897-3571}{}

\author{Alina {Wittmann}}{Technical University of Munich, Campus Straubing for Biotechnology and Sustainability, Professorship of Optimization and Sustainable Decision Making, Am Essigberg 3, 94315 Straubing, Germany \and \url{https://odm.cs.tum.de/en/} }{alina.wittmann@tum.de}{https://orcid.org/0009-0001-1314-0198}{This author's work was supported by the Deutsche Forschungsgemeinschaft (DFG, German Research Foundation) -- Project number 508981269.}

\authorrunning{E. Marica and C. Thielen and A. Wittmann} 



\nolinenumbers 

\begin{document}

\maketitle

\begin{abstract}
Given a directed, discrete-time temporal graph~$G=(V,R)$, a start node~$s\in V$, and $p\geq1$ objectives, the single-source multiobjective temporal shortest path problem asks, for each~$v\in V$, for the set of nondominated images of temporal $s$-$v$-paths together with a corresponding efficient path for each image. A recent general label setting algorithm for this problem relies on two properties of the objectives--monotonicity and isotonicity. Monotonicity generalizes the nonnegativity assumption required by label setting methods for the classical additive single-objective shortest path problem on static graphs, while isotonicity ensures that the order of the objective values of two paths is preserved when both are extended by the same arc.

In this paper, we study the problem without assuming monotonicity and/or isotonicity. A key difficulty in this setting is that zero-duration temporal cycles may need to be traversed an arbitrary finite number of times to generate all nondominated images. This motivates the study of a restricted problem variant in which a maximum admissible path length~$K$ is imposed, and only paths containing at most~$K$ arcs are considered. We develop general label correcting algorithms for this setting and establish several sufficient conditions under which such a bound is not required, implying that the algorithms compute all nondominated images.
\end{abstract}

\section{Introduction}
\label{sec:introduction}

The shortest path problem is one of the most fundamental and well-studied optimization problems in graphs and networks. In the most prominent case where the length of a path is given as the sum of the lengths of its arcs, the problem can be solved efficiently in polynomial time by label setting approaches if all arc lengths are assumed to be nonnegative or by label correcting approaches if negative lengths are allowed.

Many graphs and networks, however, are inherently dynamic in the sense that the available connections between nodes -- or even the nodes themselves -- change over time. Prominent examples include airline or railway networks, where each train or flight connection has a specific departure time~\cite{holme2012review,wu2016efficient}. Such situations are naturally modeled by (discrete-time) \emph{temporal graphs}~\cite{holme2012review,Holme+Saramäki:book}, where the nodes (corresponding to airports or railway stations in the example) are static, but the arc set changes in discrete time steps. In this case, paths of interest are so-called \emph{temporal paths} that respect the temporal nature of the graph, meaning that the next arc of a path must always start at a node after the previous one has arrived there. Given the temporal aspect, several natural interpretations of a ``shortest'' temporal path arise depending on the specific objective of interest. The most prominent objectives include minimizing the arrival time at the destination, the duration (difference between arrival time and start time), or the traversal time (sum of arc lengths) of a path, or maximizing its start time (for a given latest possible arrival time). 

In many situations, for instance when choosing a path in an airline or railway network, one might actually not only be interested in optimizing one of these objectives, but several at the same time, meaning that the considered shortest path problem in a temporal graph becomes \emph{multiobjective}. The solutions one is typically interested in when considering optimization problems with several potentially conflicting objectives are called \emph{efficient solutions} or \emph{Pareto optimal solutions}~\cite{Ehrgott:book}. An efficient solution is characterized by the property that it cannot be improved in any objective without simultaneously worsening at least one other objective. The vector of objective values (the \emph{image} of an efficient solution in the objective space) is called a \emph{nondominated image} or a \emph{nondominated point}~\cite{Ehrgott:book}. Each nondominated image corresponds to a different trade-off among the objectives, and the goal in a multiobjective optimization problem typically consists of computing all nondominated images together with one corresponding efficient solution for each of them.

These observations motivate a systematic study of the (single-source) \emph{multiobjective temporal shortest path problem}. Given a temporal graph, a source node~$s$, and $p$~objectives, this problem consists of computing, for each destination node~$v$, all nondominated images of temporal $s$-$v$-paths together with one corresponding efficient path for each image. For a very general class of objectives, including arrival time, duration, traversal time, start time, and many others, this problem has recently been studied in~\cite{Bazgan+etal:monotone} under two assumptions on the objectives. The first assumption states that the value of a path in an objective cannot improve when the path is extended by another arc (referred to as \emph{monotonicity} in the following). Monotonicity generalizes the nonnegativity assumption on arc lengths required for label setting algorithms in the classical additive shortest path problem. The second assumption, referred to as \emph{isotonicity} in the following, states that if one path has an objective value at least as good as that of another path, then this relation is preserved when both paths are extended by the same arc.
Using both assumptions, Bazgan et al.~\cite{Bazgan+etal:monotone} present a general label setting algorithm for the multiobjective temporal shortest path problem with an arbitrary number of objectives. However, similarly to the classical static shortest path problem, where arc lengths are not necessarily nonnegative, the monotonicity and isotonicity assumptions imposed in~\cite{Bazgan+etal:monotone} are not satisfied by all relevant objectives in temporal shortest path problems. 
For example, in airline networks it can often be observed that booking a flight from airport~A via airport~B to airport~C is strictly cheaper than booking only the first leg from~A to~B. Thus, when minimizing travel cost, the objective violates monotonicity. This phenomenon can be systematically exploited via dedicated websites such as Skiplagged (\url{https://skiplagged.com}) using a technique known as \emph{hidden-city ticketing}~\cite{Wang+Ye:hidden-city-ticketing}, or \emph{skiplagging}, in which a passenger books a longer itinerary but intentionally skips one or more final flight legs.

\smallskip 

In this paper, we therefore study the single-source multiobjective temporal shortest path problem with general objectives that are \emph{not} required to satisfy monotonicity and/or isotonicity. As we demonstrate through explicit examples in Section~\ref{sec:preliminaries}, a key difficulty in the case of general objectives, even for static graphs, is that cycles may need to be traversed an arbitrary finite number of times in order to generate all nondominated images. 
We therefore consider a restricted problem variant in which a maximum admissible path length~$K$ is given as part of the input, and only paths containing at most~$K$ arcs are considered. For this setting, we develop general label correcting algorithms and prove their correctness both for isotonic objectives without monotonicity (Section~\ref{sec:iso}) and for fully general objectives where neither property is assumed (Section~\ref{sec:generalAlgo}). The key difference between these two settings is that isotonicity permits the removal of dominated labels during each iteration, whereas in the absence of isotonicity dominated labels can only be safely removed after the final iteration.
In Section~\ref{sec:unrestricted_path_length}, we further establish several sufficient conditions under which no bound on the admissible path length is required. Under these conditions, our algorithms compute all nondominated images when executed for a suitably bounded number of iterations. These conditions include, in particular, rational-valued additive objectives, thereby generalizing the classical setting of additive cost objectives on static graphs.

\subsection{Related work}
\label{subsec:related_literature}
The single-source multiobjective shortest path problem has been extensively studied in static graphs (see, e.g.,~\cite{Ehrgott:book, PAIXAObook}). In the case of nonnegative additive cost objectives, the problem is known to be efficiently solvable using label setting approaches~\cite{MARTINS1984236,Ehrgott:book}. For additive objectives that may take negative values, Corley and Moon~\cite{10.1007/BF00938761} propose a multiobjective label correcting algorithm, which is also applicable in the more general setting where objectives are isotonic but not necessarily monotone~\cite{REINHARDT2011605}. Furthermore, Carraway et al.~\cite{CARRAWAY199095} present a dynamic programming approach for the one-to-one multiobjective shortest path problem in acyclic graphs without assuming monotonicity or isotonicity.

In temporal graphs, which constitute the focus of our work, the problem has received comparatively less attention. Bazgan et al.~\cite{Bazgan+etal:monotone} study the single-source multiobjective temporal shortest path problem for objectives satisfying both monotonicity and isotonicity. Even though subpaths of efficient paths may not be efficient in the temporal setting, they show how a label setting approach can still be applied for any number of isotonic, monotone objectives.
Non-monotone objectives have so far only been considered in the biobjective case, while non-isotonic objectives have, to the best of our knowledge, not been considered at all in temporal graphs. Brunelli et al.~\cite{BRUNELLI2021106086} consider a biobjective temporal shortest path problem where the first objective is the earliest arrival time, while the second objective is isotonic, but not necessarily monotone. Their approach strongly relies on the ordering of nondominated images in the biobjective case and can, therefore, not be extended to more than two objectives. Hamacher et al.~\cite{HAMACHER2006238} sketch how a label correcting approach can be applied to the biobjective temporal shortest path problem with two additive cost objectives, which are isotonic, but in general non-monotone. 

\newpage
\section{Preliminaries}
\label{sec:preliminaries}
This section contains the definitions and terminology used throughout the paper.
\begin{definition}
    A directed (discrete-time) \emph{temporal graph} is a pair~$G=(V,R)$, where~$V$ is a finite set of nodes and~$R$ is a finite set of quadruples, where each quadruple~$r\in R$ consists of a start node~$\alpha(r)\in V$, an end node~$\omega(r) \in V$, a \emph{start time}~$\tau(r) \in \mathbb{Q}_{\geq 0}$, and a \emph{traversal time}~$\lambda(r)\in \mathbb{Q}_{\geq 0}$.\footnote{Nonnegative traversal times are required for our algorithms, whereas negative start times can be handled by adding the absolute value of the smallest negative start time to all start times.}
    A quadruple~$r\in R$ is called a \emph{(temporal) arc}.
\end{definition}

Let $n \coloneqq |V|, m \coloneqq |R|$, and let $\delta^+(v)$ and $\delta^-(v)$ denote the set of outgoing and ingoing arcs of~$v\in V$, respectively. Moreover, let $[k]\coloneqq \{1,\dots,k\}$ for~$k\in\mathbb{N}_{\geq1}$ and $[0]\coloneqq \emptyset$. 
The \emph{arrival time} of an arc~$r\in R$ at its end node~$\omega(r)$ is $\tau(r) + \lambda(r)$. Therefore, a \emph{temporal path} $P=(v_0,r_1,v_1,\dots,r_k, v_k)$ is a nonempty, finite sequence of nodes and temporal arcs such that, for each $i\in [k]$, it holds that $\alpha(r_i)=v_{i-1}$ and $\omega(r_i) = v_i$, and for each $i\in [k-1]$, it holds that $\tau(r_i) + \lambda(r_i) \leq \tau(r_{i+1})$. We call a path $P=(v_0)$ with no arcs a \emph{zero-arcs path}. The \emph{length} of a path $P=(v_0, r_1,v_1,\dots,r_k, v_k)$ is the number of arcs it is composed of, i.e. $|P|=k$. A \emph{temporal $s$-$v$-path} is a temporal path $P= (v_0, r_1,v_1,\dots, r_k, v_k)$ with $v_0 = s$ and $v_k = v$, and we call the nodes~$\alpha(P)\coloneqq v_0$ and $\omega(P)\coloneqq v_k$ the \emph{start node} and the \emph{end node} of~$P$, respectively. A path $P=(v_0,r_1,v_1,\dots,r_k,v_k)$ is called \emph{simple} if the nodes $v_0,\dots,v_{k-1}$ are pairwise distinct, and it is called a \emph{(temporal) cycle} if $k\geq 1$ and $v_0 = v_k$. 
The \emph{concatenation} of two paths $P=(v_0,r_1,v_1,\dots,r_k, v_k)$ and $P'=(v'_0,r'_1,v'_1,\dots,r'_{k'}, v'_{k'})$ with $v_k = v'_0$ and $\tau(r_k)+\lambda(r_k)\leq \tau(r'_1)$ is defined as the path $(P, P')\coloneqq (v_0,r_1,v_1,\dots,r_k, v_k, r'_1,v'_1,\dots,r'_k, v'_{k'})$.
Given a path $P=(v_0,r_1,v_1,\dots,r_k, v_k)$, any path $(v_i, r_{i+1}, v_{i+1}, \dots r_j, v_j)$ with $0\leq i\leq j\leq k$ is called a \emph{subpath} of~$P$.

\smallskip

We evaluate the quality of paths using several \emph{criteria} or \emph{objectives}. The following definition generalizes the definition from~\cite{Bazgan+etal:monotone}, where nonnegative, rational-valued objectives are considered, to values in any totally ordered set. Since our analysis depends only on the structural properties of the objectives and not on a particular encoding, we do not fix a concrete representation.

\begin{definition}
    An \emph{objective} is a tuple~$(M, \leq, f, \oplus, \overline{0}, \dir)$, where~$M$ is a nonempty set, $\leq$ is a total order on~$M$, $f: R \rightarrow M$ is a function with values in~$M$, $\oplus : M \times M \rightarrow M$ is a binary operator, $\overline{0}\in M$ is a left-neutral element for $\oplus$ (i.e., $\overline{0} \oplus b = b$ for all~$b \in M$), and dir $\in \{\min, \max\}$ is a direction of optimization. Given an objective $(M, \leq, f, \oplus, \overline{0}, \dir)$, we use~$<$ to denote the strict total order corresponding to~$\leq$.
\end{definition}

For instance, maximizing the start time would be expressed by the \emph{latest start time objective}~$(\mathbb{Q}_{\geq0}\cup\{+\infty\}, \leq, \tau,\min,+\infty,\max)$, and minimizing the arrival time by the \emph{earliest arrival time objective}~$(\mathbb{Q}_{\geq0}, \leq, \lambda+\tau,\max,0,\min)$. Note that, as in~\cite{Bazgan+etal:monotone}, the binary operator $\oplus$ is required to be associative, i.e., $(a\oplus b)\oplus c = a\oplus (b\oplus c)$ for all $a,b,c \in M$, but not commutative, i.e., we can have $a\oplus b \neq b\oplus a$. Next, we define two properties of objectives, which were jointly referred to as \emph{monotonicity} in~\cite{Bazgan+etal:monotone}. Here, we define them separately for ease of discussion.

\begin{definition}\label{def:monotonicity}
    An objective $(M, \leq, f, \oplus, \overline{0}, \min)$ is \emph{monotone} if $a \leq a\oplus b$ for all~$a,b\in M$, and an objective $(M, \leq, f, \oplus, \overline{0}, \max)$ is \emph{monotone} if $a \geq a\oplus b$ for all~$a,b\in M$.
\end{definition}

\begin{definition}\label{def:isotonicity}
    An objective $(M, \leq, f, \oplus, \overline{0}, \dir)$ is \emph{isotonic} if $a \leq a'$ with $a,a'\in M$ implies that $a \oplus b \leq a' \oplus b$ for all~$b \in M$.
\end{definition}

Note that the isotonicity condition can be formulated equivalently by asking that $a \geq a'$ with $a,a'\in M$ implies that $a \oplus b \geq a' \oplus b$ for all~$b \in M$. We use both formulations in the following.  
The next definition defines objective values of paths in a temporal graph.
\begin{definition}
    Given an objective~$(M, \leq, f, \oplus, \overline{0}, \dir)$, the \emph{objective value}~$f(P)$ of a path $P=(v_0,r_1,v_1,\dots,r_k, v_k)$ is defined inductively by setting $f(P) \coloneqq \overline{0}$ for $k=0$, and $f(P) \coloneqq f(P') \oplus f(r_k)$ for $k \geq 1$, where $P' = (v_0,r_1,v_1,\dots,r_{k-1}, v_{k-1})$ is the $v_0$-$v_{k-1}$-subpath of $P$ with the last arc removed.
\end{definition}

Intuitively, monotonicity of an objective~$(M, \leq, f, \oplus, \overline{0}, \dir)$ means that extending a path~$P$ by an additional arc~$r$ cannot improve its objective value~$f(P)$. Isotonicity means that, if two temporal paths~$P_1$ and~$P_2$ with objective value~$f(P_1)$ at least as good as~$f(P_2)$ are extended by the same arc~$r$, then the objective value $f(P_1) \oplus f(r)$ of the extension of~$P_1$ will also be at least as good as the objective value $f(P_2) \oplus f(r)$ of the extension of~$P_2$.

\noindent
The following definitions introduce required terminology from multiobjective optimization.
\begin{definition}
    Let $(M_j, \leq_j, f_j, \oplus_j, \overline 0_j, \dir_j), j\in[p]$, be $p\geq 1$~objectives, and let $y=(y_1, \dots, y_p)\in M_1 \times \dots \times M_p$ and $y'=(y'_1, \dots, y'_p) \in M_1 \times \dots \times M_p$ be two vectors. The vector~$y$ is \emph{weakly dominated} by the vector $y'$ if $y'_j \leq_j y_j$, for all~$j\in [p]$ with $\dir_j = \min$, and $y_j \leq_j y'_j$ for all~$j\in [p]$ with $\dir_j = \max$. The vector~$y$ is \emph{dominated} by $y'$ if it is weakly dominated by~$y'$ and~$y \neq y'$.
\end{definition}

\begin{definition}
    Given $p\geq 1$ objectives $(M_j, \leq_j, f_j, \oplus_j, \overline 0_j,\dir_j), j\in[p]$, we call the vector $(f_1(P), \dots, \allowbreak f_p(P))\in M_1 \times \dots \times M_p$ the \emph{image} of the temporal path~$P$ and say that~$P$ is \emph{(weakly) dominated} by another temporal path~$P'$ if its image is (weakly) dominated by the image of~$P'$. If an $s$-$v$-path~$P$ is not dominated by any other $s$-$v$-path, we call~$P$ \emph{efficient} and its image \emph{nondominated}.
\end{definition}

In static shortest path problems, cycles that can be repeatedly traversed to improve an objective value are referred to as \emph{negative (cost)} or \emph{improving cycles}. In the temporal case, such cycles must necessarily be zero-duration cycles (i.e., all arcs in the cycle have traversal time zero and the same start time), and in our setting with general objectives, whether traversing a zero-duration cycle improves an objective might depend on the subpath traversed previously. This is made formal in the following definitions.

\begin{definition}\label{def:pathseq}
The \emph{duration} of a temporal path $P=(v_0,r_1,v_1,\dots,r_k, v_k)$ is $\tau(r_k)+\lambda(r_k)-\tau(r_1)$ if $k\geq1$, and zero otherwise. A cycle with duration zero is called a \emph{zero-duration cycle}. 
A zero-duration cycle $C=(v_0,r_1,v_1,\dots,r_c,v_0)$ is said to be \emph{reachable (from~$s$)} if there exists a temporal $s$-$v_0$-path $P=(v'_0,r'_1,v'_1,\dots,r'_k,v'_k)$ with $v'_k=v_0$ and $\tau(r'_k)+\lambda(r'_k)\leq\tau(r_1)$.
Given a zero-duration cycle~$C$ and a temporal path~$P$ satisfying the above condition, we define the infinite sequence of paths $(P,C^{h})_{h\geq 0}$, where~$C^h\coloneqq(C,\dots,C)$ denotes the temporal path obtained by traversing~$C$ exactly $h$~times.
\end{definition}

\begin{definition}\label{def:augmentnig_cycle}
Given a temporal graph $G=(V,R)$, a node~$s\in V$, and $p\geq1$ objectives $(M_j, \leq_j,f_j, \oplus_j, \overline 0_j,\dir_j)$, $j\in [p]$, a zero-duration cycle $C=(v_0,r_1,v_1,\dots,r_c,v_0)$ is an \emph{improving cycle in objective~$j$} if there exists a temporal $s$-$v_0$-path $P$ such that $(P,C^h)_{h\ge0}$ is well-defined (see Definition~\ref{def:pathseq}) and, for every $h\geq0$, there exists $h'>h$ with $f_j(P,C^{h'}) <_j f_j(P,C^h)$ if $\dir_j=\min$, or $f_j(P,C^{h'}) >_j f_j(P,C^h)$ if $\dir_j=\max$. We call $C$ an \emph{improving cycle} if it is improving in at least one objective~$j$.
\end{definition}

Thus, an improving cycle~$C$ in objective~$j$ is a zero-duration cycle reachable from~$s$ that can be traversed indefinitely, such that each traversal can eventually be followed by further traversals yielding strictly better values in objective~$j$. In contrast to additive cost objectives on static graphs, an improving cycle in a temporal graph need not improve the objective at every repetition. Nevertheless, it may cause the set of nondominated images of $s$-$v$-paths to be infinite or empty for some nodes~$v$.

\subsection{Problem Definition and Structural Results}\label{sec:ProblemDef}

Definition~\ref{def:augmentnig_cycle} motivates the following definition of the single-source multiobjective temporal shortest path problem, which generalizes the definition from~\cite{Bazgan+etal:monotone} to our more general definition of objectives and cases where improving cycles may exist:

\begin{definition}[Single-source multiobjective temporal shortest path problem (SSMTSPP)]\label{def:SSMTSPP}
\textbf{Instance}: A temporal graph $G=(V,R)$, a start node $s\in V$, and $p\geq 1$ objectives $(M_j,\leq_j,\allowbreak f_j, \oplus_j, \overline{0}_j, \dir)$, $j\in [p]$, the first of which is the earliest arrival time objective $(\mathbb{Q}_{\geq 0}, \leq,\allowbreak \tau+\lambda, \max, 0, \min)$.\\
\textbf{Task}: Either return the existence of an improving cycle, or compute the set of nondominated images of temporal $s$-$v$-paths for each~$v \in V$ together with a corresponding efficient path for each such image.
\end{definition}

Note that, as in~\cite{Bazgan+etal:monotone}, we assume that the first objective is the earliest arrival time objective, whose value is required to determine further arcs by which a temporal path can be extended. If this objective is not of interest, it can be excluded in a postprocessing step that returns only those images that are nondominated with respect to the other $p-1$~objectives.

In static graphs with additive cost objectives, the single-source shortest path problem is well known to be solvable using label correcting algorithms, both in the single- and multiobjective settings~\cite{e7a0cc48-246d-3ba8-9112-de6faa7aa794, P-923, 10.1007/BF00938761}. A key property underlying these algorithms is that, unless an improving cycle (negative cycle) is reachable from the start node, all nondominated images arise from simple paths of length at most~$n-1$. Consequently, it suffices to check each arc at most $n$~times, and any change in a label set during the $n$-th iteration certifies the existence of an improving cycle.
As we now demonstrate, no analogous bound on the number of iterations exists in our temporal setting with general objectives, even when all objectives are isotonic. First, positive-duration cycles may need to be traversed once in an efficient temporal path, while their positive duration prevents repeated traversal, implying that paths of length up to~$m$ must be considered in general. Moreover, the following example shows that, for general objectives, zero-duration cycles may need to be traversed an arbitrary finite--potentially exponential--number of times in order to obtain all nondominated images.

\begin{example}\label{example:exp_path_length}
    For $k\in\mathbb{N}_{\geq1}$, let $G=(V,R)$ with $V=\{s, v_1, \dots, v_k\}$, $R=\{r_1,\dots, r_{k+1}\}$ denote the temporal graph with $|V|=k+1$ nodes and $|R|=k+1$ arcs illustrated in Figure~\ref{fig:example_exp_path_length}, where $\tau(r_1) = 0$, $\lambda(r_1) = 1$, and $\tau(r_i) = 1$, $\lambda(r_i) = 0$ for $i\in\{2,\dots,k+1\}$. Let $p=2$. The first objective is the earliest arrival time objective $(\mathbb{Q}_{\geq 0}, \leq,\tau+\lambda, \max, 0, \min)$. The second objective is the isotonic but non-monotone objective $(\mathbb{Q}_{\geq 0}, \leq,f_2, \oplus_2 , 0, \min)$ defined by $f_2: R \rightarrow \mathbb{Q}_{\geq 0}$, $f_2(r_1) = k^{k+1}$, $f_2(r_i)=0$ for $i\in\{2,\dots,k\}$, and $a\oplus_2 b \coloneqq \max\{a-1, b\}$ for $a,b\in \mathbb{Q}_{\geq 0}$. 

    \smallskip

    For each node $v\in V$, there exists exactly one nondominated image of an $s$-$v$-path. For $v=s$, this image is~$(0,0)$. For each~$v\in\{v_1,\dots,v_k\}$, the unique nondominated image is~$(1,0)$. To see this, observe that after traversing~$r_1$, the first objective attains the value~$1$ and remains unchanged when the cycle~$C=(v_1, r_2, v_2,\dots ,r_{k+1},v_1)$ is traversed, since all arcs in~$C$ have zero traversal time and start time~$1$. On the other hand, each traversal of~$C$ reduces the second objective value by exactly~$k$, as long as it is positive. Starting from $f_2(r_1)=k^{k+1}$, it therefore takes $T=k^k+1$ traversals of~$C$ to reduce the second objective value at each node $v\in\{v_1,\dots,v_k\}$ to~$0$. Afterwards, further traversals of~$C$ do not change the value of either objective. 
    Consequently, the unique nondominated image at each node $v\in\{v_1,\dots,v_k\}$ is obtained by a path that traverses~$C$ exactly $T$~times. Such a path has length in $\Omega(k^{k+1})=\Omega(|R|^{|R|+1})$, which is exponential in the input size. By increasing the value of~$f_2(r_1)$, efficient paths of any finite length can be enforced. 
\end{example}

Furthermore, the following example demonstrates that non-simple cycles arising from the traversal of several non-disjoint simple cycles must also be considered.

\begin{example}\label{example:non-simple_cycle}
    Let $G=(V,R)$ with $V=\{s,v\}$ and $R=\{r_1,r_2,r_3\}$ denote the temporal graph illustrated in Figure~\ref{fig:example_non-simple_improving_cycle}, where $\tau(r_i)=\lambda(r_i)=0$ for each $i\in\{1,2,3\}$. The first objective is the earliest arrival time objective $(\mathbb{Q}_{\geq 0}, \leq, \tau+\lambda, \max, 0, \min)$. The second and third objectives are isotonic but non-monotone objectives of the form $(\mathbb{Q}\cup\{+\infty\}, \leq, f_j, \oplus_j, +\infty, \min)$ for $j\in\{2,3\}$, where the arc-value functions $f_j:R\to\mathbb{Q}\cup\{+\infty\}$ are defined by $f_2(r_1)=f_3(r_1)=f_3(r_2)=f_2(r_3)=0, f_2(r_2)=f_3(r_3)=-1$, and the binary operators satisfy $a\oplus_j b \coloneqq \min\{a,b\}$ for all $a,b\in\mathbb{Q}\cup\{+\infty\}$, $j\in\{2,3\}$.

    The only simple $s$-$s$-paths are the zero-arcs path~$P=(s)$ and the simple cycles $C_1=(s,r_1,v,r_2,s)$ and $C_2=(s,r_1,v,r_3,s)$, with images $f(P)=(0,+\infty,+\infty)$, $f(C_1)=(0,-1,0)$, and $f(C_2)=(0,0,-1)$. Repeated traversals of either~$C_1$ or~$C_2$ do not yield any further improvement in any objective. However, the image $f(C)=(0,-1,-1)$ of the non-simple cycle $C=(s,r_1,v,r_2,s,r_1,v,r_3,s)$ is nondominated and cannot be obtained by repeated traversals of any single simple cycle.
\end{example}

\begin{figure}[h]
\begin{minipage}{0.5\textwidth}
    \vspace{-0.59cm}
    \begin{tikzpicture}[>=triangle 45, auto, thick, font=\small]
    
        \coordinate (S) at (-2,0);        
        \coordinate (v1) at (0,0);
    
        \coordinate (center) at (0,1.5);
    
        \coordinate (v2) at (1.299, 2.25);   
        \coordinate (vk) at (-1.299, 2.25);  
    
        \filldraw (S) circle (2pt) node[left] {\Large  $s$};
        \filldraw (v1) circle (2pt) node[below=5pt] {\Large $v_1$};
        \filldraw (v2) circle (2pt) node[above right] {\Large $v_2$};
        \filldraw (vk) circle (2pt) node[above left] {\Large $v_k$};
    
        \draw[->] (S) -- (v1) 
            node[midway, below=3pt, align=center] {
                $\tau(r_1)=0$ \\ 
                $\lambda(r_1)=1$ \\
                $f_2(r_1)=k^{k+1}$ 
            };
    
        \draw[->] (v1) arc (-90:30:1.5)
            node[midway, right=5pt, align=center, yshift=10pt] {
                $\tau(r_2)=1$ \\ 
                $\lambda(r_2)=0$ \\
                $f_2(r_2)=0$    
            };
    
        \draw (center) ++(30:1.5) arc (30:70:1.5);
        \node at ($(center)+(90:1.5)$) {$\cdots$};
        \draw[->] (center) ++(110:1.5) arc (110:150:1.5)
            node[above=3pt, align=center] at ($(center)+(90:1.5)$) {
                $\tau(r_3) = \dots = \tau(r_{k})=1$ \\
                $\lambda(r_3) = \dots = \lambda(r_{k})=0$ \\
                $f_2(r_3) = \dots = f_2(r_{k})=0$           
            };
    
        \draw[->] (vk) arc (150:270:1.5)
            node[midway, left=5pt, align=center, yshift=10pt] {
                $\tau(r_{k+1})=1$ \\ 
                $\lambda(r_{k+1})=0$ \\
                $f_2(r_{k+1})=0$  
            };
    
    \end{tikzpicture}
    \caption{Example graph with isotonic second objective in which nondominated images arise exclusively from paths of exponential length.}\label{fig:example_exp_path_length}
\end{minipage}
\hspace{0.04\textwidth}
\begin{minipage}{0.45\textwidth}
    \centering
    \begin{tikzpicture}[>=triangle 45, thick, auto, font=\small]
        
        \coordinate (center) at (0,2);
        \coordinate (s) at (0,0);
        \coordinate (v2) at (3,0);
        
        \filldraw (s) circle (2pt) node[left] {\Large $s$};
        \filldraw (v2) circle (2pt) node[right] {\Large $v$};
        
        \filldraw (s) circle (2pt);
        \filldraw (v2) circle (2pt);
        
        \draw[->] (s) -- (v2)
            node[pos=0.5, above, sloped, align=center] {
                $\tau(r_1)=0$\\
                $\lambda(r_1)=0$
            }    node[pos=0.5, below, sloped, align=center] {
            $f_2(r_1)=0$ \\
            $f_3(r_1)=0$};
        
        \draw[->] (v2) 
            arc[start angle=0, end angle=180, radius=1.5]
            node[pos=0.5, above, sloped, align=center] {
                $\tau(r_2)=\lambda(r_2)=0$ \\ 
                $f_2(r_2)=-1$\\
                $f_3(r_2)=0$
            };
        
        \draw[->] (v2) 
            arc[start angle=0, end angle=-180, radius=1.5]
            node[pos=0.5, below, sloped, align=center] {
                $\tau(r_3)=\lambda(r_3)=0$ \\ 
                $f_2(r_3)=0$\\
                $f_3(r_3)=-1$
            };

    \end{tikzpicture}
    \caption{Example graph with isotonic second and third objectives in which a nondominated image arises exclusively from a non-simple cycle that is not a repeated traversal of a single simple cycle.
    }\label{fig:example_non-simple_improving_cycle}
\end{minipage}
\end{figure}

As demonstrated in Example~\ref{example:exp_path_length}, the length of efficient paths--and thus the number of iterations required by a label correcting algorithm to compute all nondominated images--cannot, in general, be bounded by any finite value known a priori. 
Moreover, Example~\ref{example:non-simple_cycle} illustrates that non-simple cycles arising from the traversal of several non-disjoint simple cycles must be taken into account, which can, in general, be infinitely many. We therefore also consider a restricted problem variant in which a maximum admissible path length~$K$ is given as part of the input, and only paths containing at most~$K$ arcs are considered. The following definition extends the notion of efficiency to this setting.

\begin{definition}
Given $p\geq 1$~objectives $(M_j, \leq_j, f_j, \oplus_j, \overline 0_j,\dir_j), j\in[p]$, and~$k\in\mathbb{N}$, we say that an $s$-$v$-path~$P$ of length at most~$k$ is \emph{$k$-efficient} and its image~$(f_1(P),\dots,f_p(P))$ is \emph{$k$-nondominated} if~$P$ is not dominated by any other $s$-$v$-path of length at most~$k$.
\end{definition}

\noindent
The problem variant with bounded path length is formally defined as follows:

\begin{definition}[Single-source multiobjective temporal shortest path problem with maximum path length (SSMTSPP-MPL)] \,\label{def:SSMTSPP-MPL}\\
\textbf{Instance}: A temporal graph $G=(V,R)$, a start node~$s\in V$, a maximum path length~$K$, and $p\geq 1$~objectives $(M_j, \leq_j,f_j, \oplus_j, \overline{0}_j, \dir)$, $j\in [p]$, the first of which is the earliest arrival time objective $(\mathbb{Q}_{\geq 0}, \leq,\tau+\lambda, \max, 0, \min)$.\\
\textbf{Task}: For each~$v \in V$, compute the set of $K$-nondominated images of temporal~$s$-$v$-paths together with a corresponding $K$-efficient path for each such image.
\end{definition}

In the following sections, we present label correcting algorithms for the SSMTSPP-MPL with isotonic objectives (Section~\ref{sec:iso}) and for fully general objectives (Section~\ref{sec:generalAlgo}). In Section~\ref{sec:unrestricted_path_length}, we then establish sufficient conditions under which no bound on the admissible path length is required, so that our algorithms solve the general SSMTSPP. The following definition introduces labels as used by our algorithms and extends the notion of dominance to labels.

\begin{definition}
    Given a temporal graph $G=(V,R)$ and $p\geq 1$ objectives $(M_j, \leq_j,\allowbreak f_j, \oplus_j, \overline 0_j,\dir_j)$, $j\in [p]$, a \emph{label} of a node~$v\in V$ is a $(p+3)$-tuple $(w_1,\dots,w_p, r, x, y)$, where $w_j \in M_j$, $j\in [p]$, represent the values of the objectives, $r \in \delta^{-}(v)$ and $x$ are the arc and the number of the label at the predecessor~$\alpha(r)$ from which the label has been obtained, respectively, and~$y$ is the number of the label at the current node~$v$.

    Given a label $l=(w_1, \dots, w_p, r, x, y)$, the vector $(w_1, \dots, w_p)$ is called the \emph{image corresponding to~$l$}, and~$l$ is called a \emph{corresponding label for $(w_1, \dots, w_p)$}. If $l=(w_1, \dots, w_p, r, x, y)$ and $l'=(w'_1, \dots, w'_p, r', x', y')$ are two labels of the same node~$v$, we say that~$l$ is \emph{(weakly) dominated} by~$l'$ if $(w_1, \dots, w_p)$ is (weakly) dominated by $(w'_1, \dots, w'_p)$.
\end{definition}

\section{Isotonic Objectives}
\label{sec:iso}
In this section, we study the SSMTSPP-MPL with isotonic objectives that are not required to satisfy monotonicity in the sense of Definition~\ref{def:monotonicity}; we refer to this variant as the \emph{isotonic SSMTSPP-MPL}. This setting generalizes the classical additive case in which arc costs may be negative. In that case, isotonicity holds trivially, whereas monotonicity fails whenever negative arc values are permitted.

Our label correcting algorithm for isotonic objectives is given in Algorithm~\ref{list:isotonic_labeL_correcting}. Due to isotonicity, nondominated images can arise only from extensions of nondominated labels. Hence, dominated labels can be discarded in each iteration.

\begin{lstlisting}[mathescape=true, caption={Multiobjective Temporal Label Correcting Algorithm for Isotonic Objectives},label={list:isotonic_labeL_correcting},captionpos=t,abovecaptionskip=-\medskipamount]
INPUT: temporal graph $G=(V,R)$, start node $s\in V$, maximum path
       length $K$, $p\geq 1$ isotonic objectives $(M_j, \leq_j,f_j, \oplus_j, \overline{0}_j, \dir_j)$, $j\in [p]$,
       where $(M_1, \leq_1,f_1, \oplus_1, \overline{0}_1, \dir_1)$ is the earliest arrival time objective
OUTPUT: for each $v\in V$, the set of labels corresponding to all 
        $K$-nondominated images of $s$-$v$-paths
1. initialize $L(v,k)\coloneqq [\,]$ for each $v\in V$ and each $k\in \{0,\dots,K\}$
2. add label $(\overline{0}_1, \dots, \overline{0}_p, \none, \none, 1)$ to $L(s,0)$
3. $z \coloneqq 1$
4. for $k = 0, \dots, K$ do
5.   for $v\in V$ do
6.       $L(v, k+1) = L(v, k)$
7.   for $v\in V$ with $L(v,k) \neq \emptyset$ do
8.       for $r' \in \delta^{+}(v)$ do
9.           $u \coloneqq \omega(r')$
10.          for $l = (w_1, \dots, w_p, r, x, y) \in L(v, k)$ do
11.              if $\tau(r') \geq w_1$ then
12.                  $z \coloneqq z + 1$
13.                  $l' \coloneqq (w_1\oplus_1 f_1(r'), \dots, w_p\oplus_p f_p(r'), r', y, z)$
14.              if $l'$ is not weakly dominated by any $\tilde{l} \in L(u, k+1)$ then
15.                  delete all labels in $L(u, k+1)$ dominated by $l'$
16.                  add $l'$ to $L(u, k+1)$
17.  if $L(v, k+1) = L(v,k)$ for all $v \in V$ or $k=K$ then
18.      return $L(v,k)$ for all $v \in V$
\end{lstlisting}

Observe that the earliest arrival time objective $(\mathbb{Q}_{\geq0}, \leq, \tau+\lambda, \max, 0, \min)$ is isotonic. Hence, using it as the first objective is consistent with the assumptions of Algorithm~\ref{list:isotonic_labeL_correcting}.

We begin by establishing two propositions concerning the label sets~$L(v,k)$ generated by Algorithm~\ref{list:isotonic_labeL_correcting}. Their proofs adapt the main arguments used in the correctness proof of the label correcting algorithm for the static multiobjective shortest path problem (see, e.g.,~\cite{10.1007/BF00938761}) to the temporal setting with isotonic objectives.

\begin{proposition}\label{prop:iso1}
    Suppose that Algorithm~\ref{list:isotonic_labeL_correcting} is executed without the first stopping criterion (i.e., $L(v,k+1)=L(v,k)$ for all $v\in V$) in line~17. Then, for each~$v\in V$ and each~$k \leq K$, the label set~$L(v,k)$ generated by Algorithm~\ref{list:isotonic_labeL_correcting} contains a corresponding label for each image of a $k$-efficient $s$-$v$-path. 
\end{proposition}

\begin{proof}
    Note that, without the first stopping criterion, Algorithm~\ref{list:isotonic_labeL_correcting} performs all iterations $k=0,\dots,K$ before returning the label sets~$L(v,K)$ for all~$v\in V$ in line~18. Moreover, for each $v\in V$ and~$0\leq k\leq K$, each label in~$L(v,k)$ corresponds to an image of an $s$-$v$-path of length at most~$k$. 
    
    We prove the statement by induction on~$k$. For $k=0$, the statement holds since the zero-arcs path~$P=(s)$ is the only path of length~$0$ with start node~$s$, and a corresponding label is added to~$L(s,0)$ in line~2. 

    Now suppose the statement holds for all label sets~$L(v,i)$ with~$v\in V$ and~$0\leq i \leq k$. Assume for contradiction that there exists a node~$v \in V$ and a $(k+1)$-efficient $s$-$v$-path~$P=(s, r_1,v_1,\dots,r_{\ell-1},v_{\ell-1},r_\ell, v)$, $\ell\leq k+1$, such that no label corresponding to its image $(f_1(P), \dots, f_p(P))$ is contained in~$L(v,k+1)$.
    If $\ell\leq k$, then~$P$ is $k$-efficient as well, so a label corresponding to its image is contained in~$L(v,k)$ by induction and, thus, added to~$L(v,k+1)$ in line~6 of iteration~$k$. By $(k+1)$-efficiency of~$P$, this label cannot be deleted from~$L(v,k+1)$ afterwards, which contradicts the choice of~$P$. Thus, we may assume in the following that~$\ell=k+1$.
    Hence, the $s$-$v_{\ell-1}$-path~$P'\coloneqq(s, r_1,v_1,\dots,r_{\ell-1}, v_{\ell-1})$ has length~$\ell-1=k\geq 0$. If there exists a label in~$L(v_{\ell-1},k)$ that corresponds to the image of~$P'$, then, as~$r_\ell\in\delta^+(v_{\ell-1})$ and~$\tau(r_{\ell})\geq \tau(r_{\ell-1})+\lambda(r_{\ell-1})$, a label corresponding to the image of~$P$ would be created and, since~$P$ is $(k+1)$-efficient, this label would be added to and kept in~$L(v,k+1)$, which contradicts the choice of~$P$. Thus, we can assume that~$L(v_{\ell-1},k)$ does not contain a label corresponding to the image of~$P'$. By induction hypothesis, this implies that~$P'$ must be dominated by some $k$-efficient $s$-$v_{\ell-1}$-path~$P''$, and a label~$l''$ corresponding to the image of~$P''$ must be contained in~$L(v_{\ell-1},k)$. Moreover, since the first objective is the earliest arrival time objective, the dominating path~$P''$ must arrive at~$v_{\ell-1}$ no later than~$P'$, which means that~$(P'',r_\ell,v)$ is a temporal $s$-$v$-path for which a corresponding label~$l$ would have been created when extending the label~$l''$ using the arc~$r_\ell$. But since~$P''$ dominates~$P'$, isotonicity of the objectives implies that~$(P'',r_\ell,v)$ weakly dominates~$P=(P',r_\ell,v)$. Since~$(P'',r_\ell,v)$ has length at most~$k+1$ and~$P$ is $(k+1)$-efficient, this implies that~$(P'',r_\ell,v)$ and~$P$ have the same image, for which~$l$ is a corresponding label. By $(k+1)$-efficiency of~$P$, this label would have been added to and kept in the label set~$L(v,k+1)$, which contradicts the choice of~$P$.
\end{proof}

\begin{proposition}\label{prop:iso2}
    Suppose that Algorithm~\ref{list:isotonic_labeL_correcting} is executed without the first stopping criterion (i.e., $L(v,k+1)=L(v,k)$ for all $v\in V$) in line~17. Then, for each~$v\in V$ and each~$k \leq K$, each label in the set~$L(v,k)$ generated by Algorithm~\ref{list:isotonic_labeL_correcting} corresponds to an image of a $k$-efficient $s$-$v$-path.    
\end{proposition}

\begin{proof}
    For $k=0$, the claim holds since the only label added to any label set~$L(v,0)$ is the initial label $(\overline{0}_1, \dots, \overline{0}_p, \none, \none, 1)$ added to $L(s,0)$ in line~2, which corresponds to the image of the $0$-efficient zero-arcs path~$P=(s)$ (the only path of length~$0$ starting at~$s$). So now, let $k \in [K]$ and let $l=(w_1, \dots, w_p, r, x, y)\in L(v,k)$ be a label for some~$v\in V$. We inductively define~$r_k \coloneqq r$ as the arc from which label~$l$ has been obtained, $r_{k-1}$ as the arc from which the $x$-th label in~$L(\alpha(r_k),k-1)$ has been obtained, and so forth until the label $(\overline{0}_1, \dots, \overline{0}_p, \none, \none, 1)\in L(s,0)$ is obtained, which must happen after at most~$k$ steps since the second parameter of the label set decreases by one in each step. By construction, the image of the resulting $s$-$v$-path~$P=(s,r_{k-j},v_{k-j}\dots,r_k, v)$ corresponds to the label~$l$. The path is temporal since, due to line~11 of the algorithm, labels are only extended using further arcs if the first component of the label corresponding to the arrival time at the current node is no larger than the start time of the arc.
    
    It remains to show that~$P$ is not dominated by any other $s$-$v$-path of length at most~$k$. Assume for a contradiction that~$P$ is dominated by such a path~$P'$. Then, by iteratively replacing~$P'$ by a path of length at most~$k$ that dominates it until no such path exists anymore, we can assume that~$P'$ is $k$-efficient. Thus, by Proposition~\ref{prop:iso1}, the set~$L(v,k)$ generated by the algorithm contains a label~$l'$ corresponding to the image of~$P'$. Since the label~$l'$ dominates the label~$l$ corresponding to~$P$ by construction, this means that~$l$ would either have not been added to~$L(v,k)$ if~$l'$ was already in~$L(v,k)$ when~$l$ was created, or would have been deleted from~$L(v,k)$ when~$l'$ was added. This yields the desired contradiction.
\end{proof}

\noindent
The following theorem establishes the correctness of Algorithm~\ref{list:isotonic_labeL_correcting}.

\begin{theorem}\label{theorem:iso_main_th}
    For each node~$v \in V$, the set~$L(v,k)$ returned by Algorithm~\ref{list:isotonic_labeL_correcting} contains a corresponding label for each image of a $K$-efficient $s$-$v$-path, and each label in~$L(v,k)$ corresponds to an image of this type.
\end{theorem}

\begin{proof}
    If Algorithm~\ref{list:isotonic_labeL_correcting} is executed without the first stopping criterion in line~17, the claim follows directly from Propositions~\ref{prop:iso1} and~\ref{prop:iso2}. It therefore remains to justify the correctness of the first stopping criterion. Suppose that, in some iteration $k<K$, the condition $L(v,k+1)=L(v,k)$ holds for all~$v\in V$. Then no label set changes in iteration~$k$. Since subsequent iterations depend only on the label sets produced in the previous iteration, it follows inductively that no further changes occur in any iteration $k+1,\dots,K$. Hence, executing the algorithm without the first stopping criterion would yield $L(v,k)=L(v,K)$ for all~$v\in V$. Consequently, terminating the algorithm already in iteration~$k$ and returning the sets~$L(v,k)$ for all~$v\in V$ yields the correct output as well, which completes the proof.
\end{proof}

After termination of Algorithm~\ref{list:isotonic_labeL_correcting}, the $K$-nondominated images corresponding to the labels in the returned sets~$L(v,k)$, $v\in V$, can be directly obtained as the vectors of the first $p$~components of these labels and, for each such image, a corresponding path can be obtained by using the label entries providing the predecessor arcs and label numbers. Moreover, the proofs imply the following important corollary. It will form the basis for some of the sufficient conditions provided in Section~\ref{sec:unrestricted_path_length} under which Algorithm~\ref{list:isotonic_labeL_correcting} actually solves the general SSMTSPP.

\begin{corollary}\label{cor:unchanged_labels}
Suppose that Algorithm~\ref{list:isotonic_labeL_correcting} terminates in iteration~$k$ because the stopping criterion $L(v,k+1)=L(v,k)$ holds for all~$v\in V$ in line~17. Then, for each~$v\in V$, the returned set~$L(v,k)$ contains a corresponding label for each image of an efficient $s$-$v$-path, and each label in~$L(v,k)$ corresponds to an image of this type.
\end{corollary}

Furthermore, Proposition~\ref{prop:iso2} directly yields the following result concerning the size of the label sets created during the algorithm:

\begin{corollary}\label{cor:iso_runtime}
    In each iteration~$k$ of Algorithm~\ref{list:isotonic_labeL_correcting}, the size of each label set~$L(v,k)$ is bounded by the number of $k$-nondominated images of $s$-$v$-paths.
\end{corollary}

\section{General Objectives}
\label{sec:generalAlgo}

In this section, we present a label correcting algorithm (Algorithm~\ref{list:general_labeL_correcting}) for the SSMTSPP-MPL in the general case where the objectives are neither required to be monotone nor isotonic. In the absence of isotonicity, labels corresponding to nondominated images may arise exclusively from extending dominated labels. Therefore, dominated labels cannot be discarded in earlier iterations, in contrast to Algorithm~\ref{list:isotonic_labeL_correcting}. 

\begin{lstlisting}[mathescape=true, caption={General Multiobjective Temporal Label Correcting Algorithm } ,label={list:general_labeL_correcting},captionpos=t,abovecaptionskip=-\medskipamount]
INPUT: temporal graph $G=(V,R)$, start node $s\in V$, maximum path
       length $K$, $p\geq 1$ objectives $(M_j, \leq_j, f_j, \oplus_j, \overline{0}_j, \dir_j), j\in [p]$, where
       $(M_1, \leq_1,f_1, \oplus_1, \overline{0}_1, \dir_1)$ is the earliest arrival time objective
OUTPUT: for each $v\in V$, the set of labels corresponding to all 
        $K$-nondominated images of $s$-$v$-paths
1. initialize $L(v,k)\coloneqq[\,]$ for each $v\in V$ and each $k\in \{0,\dots,K\}$
2. add label $(\overline{0}_1, \dots, \overline{0}_p, \none, \none, 1)$ to $L(s,0)$
3. $z \coloneq 1$
4. for $k = 0, \dots, K$ do
5.    for $v\in V$ do
6.        $L(v, k+1) = L(v, k)$
7.    for $v\in V$ with $L(v,k) \neq \emptyset$ do
8.        for $r' \in \delta^{+}(v)$ do
9.            $u \coloneqq \omega(r')$
10.           for $l = (w_1, \dots, w_p, r, x, y) \in L(v, k)$ do
11.               if $\tau(r') \geq w_1$ then
12.                   $z \coloneqq z + 1$
13.                   $l' \coloneqq (w_1\oplus_1 f_1(r'), \dots, w_p\oplus_p f_p(r'), r', y, z)$
14.               if $\nexists \tilde{l} \in L(u, k+1)$ with the same image as $l'$ then
15.                   add $l'$ to $L(u, k+1)$
16.   if $L(v,k+1) = L(v,k)$ for all $v \in V$ or $k=K$ then
17.       for $v \in V$ do
18.           delete all dominated labels in $L(v, k)$
19.       return $L(v,k)$ for all $v \in V$
\end{lstlisting}

The proof of correctness of Algorithm~\ref{list:general_labeL_correcting} is similar to the proofs of Proposition~\ref{prop:iso1} and Theorem~\ref{theorem:iso_main_th}.
The key difference is that dominated labels are removed only at the end of the final iteration~$k'\leq K$, immediately before termination. 
Consequently, in all iterations~$k<k'$, each set~$L(v,k)$ generally contains not only labels corresponding to images of $k$-efficient paths, but one corresponding label for each possible image of an $s$-$v$-path of length at most~$k$.

\begin{proposition}\label{prop:general_L}
    Suppose that Algorithm~\ref{list:general_labeL_correcting} is executed without the first stopping criterion (i.e. $L(v, k+1) = L(v,k)$ for all $v\in V$) in line 16. Then, for each~$v\in V$ and each~$k\leq K$, the label set~$L(v,k)$ right before dominated labels are deleted in line~18 of the final iteration contains one corresponding label for each image of an $s$-$v$-path of length at most~$k$.
\end{proposition}
\begin{proof} 
    As noted for Algorithm~\ref{list:isotonic_labeL_correcting} in the proof of Proposition~\ref{prop:iso1}, also Algorithm~\ref{list:general_labeL_correcting} has the property that each label in the generated set~$L(v,k)$ for~$v\in V$ and~$k\leq K$ corresponds to an image of an $s$-$v$-path of length at most~$k$.

    We now prove the statement by induction on~$k$. 
    For $k=0$, the statement holds since the zero-arcs path~$P=(s)$ is the only path of length~$0$ with start node~$s$, and a corresponding label is added to~$L(s,0)$ in line~2. 

    Now suppose the statement holds for all label sets~$L(v,i)$ with~$v\in V$ and $0\leq i\leq k$. Assume for contradiction that there exists a node~$v\in V$ and an $s$-$v$-path~$P=(s,r_1,v_1,\dots,r_{\ell-1},v_{\ell-1},r_\ell,v)$ of length~$\ell\leq k+1$ such that no label corresponding to its image is contained in~$L(v,k+1)$ right before dominated labels are deleted in line~18 of the final iteration. If~$P$ has length at most~$k$, then, by induction, the set~$L(v,k)$ contains a label~$l$ corresponding to its image, and~$l$ would have been copied to~$L(v,k+1)$ in line~6 of iteration~$k$. Thus, we may assume in the following that~$P$ has length~$k+1$. Hence, the $s$-$v_{\ell-1}$-path $P'\coloneqq (s,r_1,v_1,\dots,r_{\ell-1},v_{\ell-1})$ has length $\ell-1=k\geq0$. 
    Thus, by induction, there exists a label in~$L(v_{\ell-1},k)$ that corresponds to the image of~$P'$. Since $r_{\ell}\in\delta^+(v_{\ell-1})$ and~$\tau(r_{\ell})\geq \tau(r_{\ell-1})+\lambda(r_{\ell-1})$, a label~$l$ corresponding to the image of~$P$ is created by extending this label, and~$l$ is added to~$L(v,k+1)$. This contradicts the choice of~$P$ and finishes the proof.
\end{proof}

\noindent
The following theorem establishes the correctness of Algorithm \ref{list:general_labeL_correcting}.

\begin{theorem}\label{theorem:general_main}
    For each node~$v\in V$, the label set~$L(v,k)$ returned by Algorithm~\ref{list:general_labeL_correcting} contains a corresponding label for each image of a $K$-efficient $s$-$v$-path, and each label in~$L(v,k)$ corresponds to an image of this type.
\end{theorem}
\begin{proof}
    If Algorithm~\ref{list:general_labeL_correcting} is executed without the first stopping criterion in line~16, Proposition~\ref{prop:general_L} shows that, for each~$v\in V$, the label set~$L(v,k)$ right before dominated labels are deleted in line~18 of the final iteration~$k$ contains one corresponding label for each image of an $s$-$v$-path of length at most~$k$. Since exactly the dominated labels are deleted from~$L(v,k)$ in line~18 before returning the set, the claim holds when the first stopping criterion in line 16 is omitted. It therefore remains to justify the correctness of this stopping criterion. Suppose that, in some iteration $k<K$, the condition $L(v,k+1)=L(v,k)$ holds for all~$v\in V$. Then no new labels are added to any label set in iteration~$k$. Since subsequent iterations depend only on the label sets produced in the previous iteration, it follows inductively that no new labels are added to any label set in iterations $k+1,\dots K$. Hence, executing the algorithm without the first stopping criterion would yield $L(v,k)=L(v,K)$ for all~$v\in V$ before the deletion of dominated labels from~$L(v,K)$ in the final iteration. Consequently, terminating the algorithm already in iteration~$k$ and deleting dominated labels from~$L(v,k)$ instead yields the correct output as well, which completes the proof.
\end{proof}

Since dominated labels are removed only at the end of the final iteration, the size of~$L(v,k)$ can no longer be bounded by the number of $k$-nondominated images of $s$-$v$-paths as in Corollary~\ref{cor:iso_runtime}; in fact, it may be exponential in this number. Nevertheless, by arguments analogous to those for Algorithm~\ref{list:isotonic_labeL_correcting}, the following statement holds for Algorithm~\ref{list:general_labeL_correcting}:

\begin{corollary}\label{cor:unchanged_labels_2}
Suppose that Algorithm~\ref{list:general_labeL_correcting} terminates in iteration~$k$ because the stopping criterion $L(v,k+1)=L(v,k)$ holds for all~$v\in V$ in line~16. Then, for each~$v\in V$, the returned set~$L(v,k)$ contains a corresponding label for each image of an efficient $s$-$v$-path, and each label in~$L(v,k)$ corresponds to an image of this type.
\end{corollary}

Note that, for monotone but non-isotonic objectives, the same general algorithm is required since monotonicity alone does not justify discarding dominated labels in earlier iterations. 
Hence, all possible images must be retained until immediately before termination.

\section{The SSMTSPP with Unrestricted Path Lengths}\label{sec:unrestricted_path_length}

Example~\ref{example:exp_path_length} in Section~\ref{sec:ProblemDef} shows that temporal paths of arbitrarily large finite length may need to be considered in the SSMTSPP, even when all objectives are isotonic. Consequently, Algorithms~\ref{list:isotonic_labeL_correcting} and~\ref{list:general_labeL_correcting} can, in general, solve the problem only when a maximum admissible path length~$K$ is imposed. In Section~\ref{sec:sufficient_conds}, we establish several sufficient conditions under which no such bound is required. Under these conditions, Algorithm~\ref{list:isotonic_labeL_correcting} correctly solves the SSMTSPP for isotonic objectives and Algorithm~\ref{list:general_labeL_correcting} does so for arbitrary objectives when executed for a sufficient, precomputable number of iterations. Section~\ref{sec:additive_objectives} then considers the specific isotonic setting of \emph{rational additive objectives}, which generalizes the classical additive cost setting, and shows that executing Algorithm~\ref{list:isotonic_labeL_correcting} for~$m=|R|$ iterations suffices to solve the SSMTSPP.

\subsection{General Sufficient Conditions}\label{sec:sufficient_conds}

A first sufficient condition under which Algorithms~\ref{list:isotonic_labeL_correcting} and~\ref{list:general_labeL_correcting} find labels for all images of efficient paths of arbitrary length relates to the existence of reachable zero-duration cycles:

\begin{theorem}\label{theorem:no_cycles}
Suppose that the temporal graph~$G$ contains no zero-duration cycles reachable from~$s$.  
Then, for isotonic objectives, executing Algorithm~\ref{list:isotonic_labeL_correcting} for $m$~iterations produces, for each $v\in V$, a label set~$L(v,k)$ such that
\begin{enumerate}
    \item for every efficient $s$-$v$-path, there exists a label in~$L(v,k)$ corresponding to its image, and
    \item every label in~$L(v,k)$ corresponds to the image of some efficient $s$-$v$-path.
\end{enumerate}
Moreover, Algorithm~\ref{list:general_labeL_correcting} with $K\coloneqq m$ achieves the same guarantee for arbitrary objectives.
\end{theorem}

\begin{proof}
    It suffices to show that no $s$-$v$-path of length greater than~$m$ exists for any node~$v\in V$. Any such path would traverse some temporal arc~$r\in R$ more than once, implying the existence of a zero-duration cycle reachable from~$s$, namely the $\alpha(r)$–$\alpha(r)$-subpath between two consecutive traversals of~$r$. This contradicts the assumption, and the claim follows.
\end{proof}

Note that the condition that the temporal graph~$G$ contains no zero-duration cycles reachable from~$s$ can be verified in polynomial time.\footnote{First, compute the earliest arrival time for every node~$v\in V$ with the approach from~\cite{wu2016efficient}. Then, for each timestamp~$t$ for which an arc~$r\in R$ with $\tau(r)=t$ and $\lambda(r)=0$ exists, construct the subgraph induced by all nodes whose earliest arrival time is at most~$t$ and by all zero traversal time arcs with start time~$t$. There are at most $m$ such subgraphs, each computable in polynomial time. Finally, perform a depth-first search in each subgraph to check whether it contains a cycle.}

Another case where the argumentation from the above proof can be used is when a positive \emph{minimum waiting time}~$\Delta(v)\in\mathbb{Q}_{>0}$ is imposed at each node~$v\in V$, a constraint that is common in temporal network settings (see, e.g.,~\cite{bentert2020efficient}). In this case, the condition that $\tau(r_i) + \lambda(r_i) \leq \tau(r_{i+1})$ for each~$i\in [k-1]$ must hold for a temporal path~$P=(v_0,r_1,v_1,\dots,r_k, v_k)$ is replaced by $\tau(r_i) + \lambda(r_i) \leq \tau(r_{i+1})-\Delta(v_i)$ for each~$i\in [k-1]$ to enforce the minimum waiting time at each intermediate node of the path. This condition can be incorporated directly into Algorithms~\ref{list:isotonic_labeL_correcting} and~\ref{list:general_labeL_correcting} by replacing line~11 by the following to ensure that labels are only extended if the required minimum waiting time is achieved at the current node~$v$:

\begin{lstlisting}[mathescape=true, caption={} ,label={list:waiting_time_constraints},captionpos=t,abovecaptionskip=-\medskipamount]
11.               if $\tau(r') \geq w_1 + \Delta(v)$ then
\end{lstlisting}

It is easy to see that the proofs of correctness for both algorithms generalize directly to the modified versions. Moreover, no arc can ever be traversed twice in the presence of a positive minimum waiting time at each node, so we obtain the following corollary:

\begin{corollary}\label{cor:waiting_time_constraints}
Suppose that a positive minimum waiting time~$\Delta(v)\in\mathbb{Q}_{>0}$ is imposed at each node~$v\in V$.  
Then, for isotonic objectives, executing Algorithm~\ref{list:isotonic_labeL_correcting} for $m$~iterations produces, for each $v\in V$, a label set~$L(v,k)$ such that
\begin{enumerate}
    \item for every efficient $s$-$v$-path, there exists a label in~$L(v,k)$ corresponding to its image, and
    \item every label in~$L(v,k)$ corresponds to the image of some efficient $s$-$v$-path.
\end{enumerate}
Moreover, Algorithm~\ref{list:general_labeL_correcting} with $K\coloneqq m$ achieves the same guarantee for arbitrary objectives.
\end{corollary}

Note that, in the situations described in Theorem~\ref{theorem:no_cycles} and Corollary~\ref{cor:waiting_time_constraints}, a nondominated image may arise exclusively from paths of length greater than~$n$, as a positive-duration temporal cycle might need to be traversed once to obtain it. Consequently, in contrast to the static case with additive cost objectives, $m$ rather than~$n$ must be used as the maximum number of iterations in the algorithms.

While the previous results use the temporal graph structure to bound the required number of iterations of the algorithms, the following proposition instead relies on the number of values that the objectives can take.

\begin{restatable}{proposition}{PropositionKDistinctImages}\label{prop:K_distinct_images}
Suppose that, for each node~$v\in V$, at most~$\kappa$ distinct images of $s$-$v$-paths exist.  
Then, for isotonic objectives, the stopping criterion $L(v,k+1)=L(v,k)$ for all~$v\in V$ in line~17 of Algorithm~\ref{list:isotonic_labeL_correcting} is satisfied after at most~$m\kappa$ iterations and the algorithm produces, for each $v\in V$, a label set~$L(v,k)$ such that
\begin{enumerate}
    \item for every efficient $s$-$v$-path, there exists a label in~$L(v,k)$ corresponding to its image, and
    \item every label in~$L(v,k)$ corresponds to the image of some efficient $s$-$v$-path.
\end{enumerate}
Moreover, Algorithm~\ref{list:general_labeL_correcting} satisfies the stopping criterion $L(v,k+1)=L(v,k)$ for all~$v\in V$ in line~16 after at most~$m\kappa$ iterations and achieves the same guarantee for arbitrary objectives.
\end{restatable}

\begin{proof}
    By Corollaries~\ref{cor:unchanged_labels} and~\ref{cor:unchanged_labels_2}, it suffices to show that the stopping criterion is satisfied after at most~$m\kappa$ iterations. 
    We first establish this claim for Algorithm~\ref{list:isotonic_labeL_correcting} and then explain how the argument extends to Algorithm~\ref{list:general_labeL_correcting}.
    
    Assume for a contradiction that the stopping criterion $L(v,k+1)=L(v,k)$ for all~$v\in V$ is not satisfied in Algorithm~\ref{list:isotonic_labeL_correcting} after at most~$m\kappa$ iterations.
    Then $L(v,m\kappa+1)\neq L(v,m\kappa)$ must hold for some node~$v\in V$ in iteration~$m\kappa$. Thus, there is either a label~$l \in L(v,m\kappa+1)$ that is not in $L(v, m\kappa)$, or there is a label~$l' \in L(v, m\kappa)$ that is not in~$L(v,m\kappa+1)$. The latter implies that~$l'$ was removed in iteration~$m\kappa$ because a new label dominating it was created. Therefore, in both cases, a new label~$l$ is added to~$L(v,m\kappa+1)$ in iteration~$m\kappa$. By Proposition~\ref{prop:iso2}, this can only be the case if no label weakly dominating~$l$ arises from any $s$-$v$-path of length at most~$m\kappa$. Hence, the $s$-$v$-path~$P$ obtained from~$l$ using the predecessor arcs and label numbers as in the proof of Proposition~\ref{prop:iso2} has length~$|P|>m\kappa$. Since~$|P|>m\kappa$, there must be at least one arc~$r\in R$ that was used $\kappa(r)\geq \kappa+1$ times in~$P$. We let $P=(P_1,\dots,P_{\kappa(r)},P')$ denote the partition of~$P$ into subpaths such that the subpath~$P_i$, $i\in[\kappa(r)]$, ends at node~$\omega(r)$ right after the~$i$th traversal of arc~$r$. Then all the $\kappa(r)\geq \kappa+1$ many $s$-$\omega(r)$-paths~$(P_1,\dots,P_i)$, $i\in[\kappa(r)]$, must have pairwise distinct images since, otherwise, omitting at least one of the nonempty subpaths~$P_i$ in~$P=(P_1,\dots,P_{\kappa(r)},P')$ would have resulted in a path of length at most~$m\kappa$ with the same image as~$P$. This contradicts the assumption that there exist at most~$\kappa$ different images of $s$-$v$-paths.

    \smallskip

    The argument for Algorithm~\ref{list:general_labeL_correcting} is analogous. Suppose that $L(v,m\kappa+1)\neq L(v,m\kappa)$ for some node~$v\in V$. Then at least one new label~$l\notin L(v,m\kappa)$ must have been created for~$v$ in iteration~$m\kappa$, and~$l$ does not arise from any $s$-$v$-path of length at most~$m\kappa$. The remainder of the proof is identical to the argument given above for Algorithm~\ref{list:isotonic_labeL_correcting}.
\end{proof}

A bound~$\kappa$ on the number of distinct images at each node, as required in Proposition~\ref{prop:K_distinct_images}, can in particular be obtained if, for each objective~$j\in[p]$, an upper bound $\kappa(j)\in\mathbb{N}_{\geq1}$ on the number of values that objective~$j$ can attain at a node is known, which yields:

\begin{corollary}\label{cor:bounding_the-values}
Suppose that, for each objective~$j\in[p]$ and each node~$v\in V$, the number of values that objective~$j$ can attain over all temporal $s$-$v$-paths is bounded by $\kappa(j)\in\mathbb{N}_{\geq1}$. Then, with $\kappa\coloneqq\prod_{j=1}^p \kappa(j)$, the stopping criterion $L(v,k+1)=L(v,k)$ for all~$v\in V$ is satisfied after at most $m\kappa$~iterations in line~17 of Algorithm~\ref{list:isotonic_labeL_correcting} and in line~16 of Algorithm~\ref{list:general_labeL_correcting}. Consequently, both algorithms achieve the guarantee stated in Proposition~\ref{prop:K_distinct_images}.
\end{corollary}

Corollary~\ref{cor:bounding_the-values} applies immediately if the set~$M_j$ is finite for each objective $(M_j, \leq_j, f_j, \oplus_j,\allowbreak \overline{0}_j, \dir_j)$. For example, this holds for the objective $(\{-1,0,1\}, \leq, f, \cdot, 1, \min)$, which is neither isotonic nor monotone. More generally, it suffices that the closure $\cl(f_j(R)) \subseteq M_j$ of the arc values under~$\oplus_j$ is finite. As an illustration, consider the non-isotonic objective $(\mathbb{N}, \leq,\allowbreak f, \oplus, 0, \max)$ defined by $a \oplus b \coloneqq b$ if $a=0$ or $a=b$, and $a \oplus b \coloneqq 0$ otherwise, for $a,b \in \mathbb{N}$. Interpreting each value as a category, maximizing the objective corresponds to finding a path whose arcs share the same, largest possible category. In this case, $|\cl(f_j(R))| \leq m+1$.

A related application of Corollary~\ref{cor:bounding_the-values} is based on the observation that, for each objective~$j\in[p]$ and node~$v\in V$, the number of attainable values is bounded by $(|\delta^-(v)|+1)\leq m+1$ times the maximum number attainable by $s$-$v$-paths sharing the same last arc, where the additional~$+1$ accounts for the zero-arcs path when $s=v$. An analogous bound using $(|\delta^+(s)|+1)\leq m+1$ holds for paths sharing the same first arc. Since the arrival time of a path depends only on its last arc and the start time only on its first arc, it follows that for both the earliest arrival time objective $(\mathbb{Q}_{\geq 0}, \leq, \tau+\lambda, \max, 0, \min)$ and the latest start time objective $(\mathbb{Q}_{\geq 0}\cup\{+\infty\}, \leq, \tau, \min, +\infty, \max)$, the number of attainable values at any node~$v$ is bounded by~$m+1$. 
Similar reasoning applies to many objectives that are neither isotonic, monotone, nor commutative. For instance, consider the objective $(\mathbb{Z}\cup\{-\infty\}, \leq, f, \oplus, -\infty, \min)$ defined by $a \oplus b \coloneqq b$ if $a \leq b$, and $a \oplus b \coloneqq 0$ otherwise. Every path ending with an arc~$r\in R$ can then attain only the two values~$0$ and~$f(r)$, implying that the number of attainable values at any node~$v$ is bounded by~$2m+1$.

\subsection{Rational Additive Objectives}\label{sec:additive_objectives}

In this subsection, we consider objectives that are additive over the rational numbers.

\begin{definition}\label{def:additive_objective}
An objective $(\mathbb{Q}, \leq, f, +, 0, \dir)$, where $+$ denotes addition on $\mathbb{Q}$, $\leq$ the standard order, and $\dir\in\{\min,\max\}$, is called \emph{rational additive}.
\end{definition}

We refer to the variant of the SSMTSPP in which the first objective is the earliest arrival time objective and the remaining $p-1$ objectives are rational additive as the \emph{additive SSMTSPP}. Rational additive objectives are clearly isotonic.

Based on this structure, we consider a modified version of 
Algorithm~\ref{list:isotonic_labeL_correcting}. 
The algorithm is executed with $K \coloneqq m$ and differs only in its termination rule, 
which is replaced as follows in order to detect improving cycles.

\begin{lstlisting}[mathescape=true, caption={Multiobjective Temporal Label Correcting Algorithm for Additive Objectives},label={list:alg2},captionpos=t,abovecaptionskip=-\medskipamount]
17.   if $L(v, k+1) = L(v,k)$ for all $v \in V$ then
18.       return $L(v,k)$ for all $v \in V$
19.   if $k=m$ then
20.       return ``an improving cycle exists''
\end{lstlisting}

The same arguments as for Algorithm~\ref{list:isotonic_labeL_correcting} show that Propositions~\ref{prop:iso1} and~\ref{prop:iso2} remain valid for Algorithm~\ref{list:alg2} with~$K \coloneqq m$. 
Moreover, Corollary~\ref{cor:unchanged_labels} applies as well. Hence, if Algorithm~\ref{list:alg2} terminates in line~18, then, for each~$v\in V$, the returned label set~$L(v,k)$ contains a corresponding label for each image of an efficient $s$-$v$-path and only such labels. Based on these observations, the following two propositions show that Algorithm~\ref{list:alg2} either computes all nondominated images of paths of arbitrary length or correctly detects an improving cycle.

\begin{restatable}{proposition}{propstrictlyisocycle}\label{prop:3}
    If the temporal graph~$G$ contains an improving cycle, then Algorithm~\ref{list:alg2} terminates in line~20 and correctly reports its existence.
\end{restatable}
\begin{proof}
    Without loss of generality, assume that $\dir_j=\min$ for all~$j\in[p]$. Let~$v_0\in V$ be any node such that an improving cycle with start node~$v_0$ exists. 
    Assume for a contradiction that Algorithm~\ref{list:alg2} does not terminate in line~20. Then, the algorithm must terminate in line~18, so $L(v,k+1)=L(v,k)$ holds for all~$v\in V$ in line~17 of the final iteration~$k\leq m$. Moreover, for each~$v\in V$, the returned label set~$L(v,k)$ contains a corresponding label for each image of an efficient $s$-$v$-path and only such labels. 

    We first show that the returned set~$L(v_0,k)$ must be nonempty. 
    Let~$k'$ denote the minimum length of an $s$-$v_0$-path, and let~$P'$ denote some $s$-$v_0$-path of length~$k'$. Then, by the arguments used in the proof of Proposition~\ref{prop:iso2}, a label~$l$ whose corresponding image weakly dominates~$f(P')$ would be contained in~$L(v_0,k')$ after executing the algorithm for $k'$~iterations. Thus, since no label sets can change after iteration~$k$ because $L(v,k)=L(v,k+1)$ for all~$v\in V$ holds at the end of iteration~$k$, such a label must also be contained in~$L(v_0,k)$, which shows that~$L(v_0,k)$ is nonempty.
    We now distinguish two cases:
    
    \textbf{Case 1: {\normalfont There exists an improving cycle $C=(v_0,r_1,v_1,\dots,r_c,v_0)$ that satisfies $f_j(C) \leq 0$ for all~$j\in[p]$.}}
    
    Let~$l\in L(v_0,k)$ be any nondominated label and let~$P$ denote the corresponding efficient $s$-$v_0$-path obtained from~$l$ using the predecessor arcs and label numbers as in the proof of Proposition~\ref{prop:iso2}. Since~$C$ is improving, we must have~$f_j(C)<0$ for at least one objective~$j$. Therefore, the path~$(P,C)$ dominates the path~$P$, which contradicts the efficiency of~$P$. 

    \textbf{Case 2: {\normalfont For each improving cycle $C=(v_0,r_1,v_1,\dots,r_c,v_0)$, there exists an objective $j'\in [p]$ such that $f_{j'}(C)>0$.}}
    
    If~$C$ is any such improving cycle, we must have~$f_j(C)<0$ for at least one objective~$j\neq j'$. Thus, each traversal of~$C$ strictly decreases objective~$j$ while strictly increasing objective~$j'$, which implies that there exist infinitely many nondominated images of $s$-$v_0$-paths. This contradicts the fact that the finite label set~$L(v_0,k)$ returned by the algorithm contains a corresponding label for each image of an efficient $s$-$v$-path.
\end{proof}

\noindent
The following lemma is required for the proof of Proposition~\ref{prop:5}.
\begin{lemma}\label{lemma:4}
    If Algorithm~\ref{list:alg2} terminates in line~20, then the temporal graph~$G$ contains an improving cycle.
\end{lemma}
\begin{proof}
    If the algorithm returns the existence of an improving cycle in line~20, this implies that $k=m$ and $L(v, m+1) \neq L(v, m)$ for at least one node~$v\in V$. Thus, there is either a label~$l\in L(v,m+1)$ that is not in~$L(v, m)$, or there is a label~$l'\in L(v,m)$ that is not in~$L(v,m+1)$. The latter implies that~$l'$ was removed in iteration~$m$ because a new label dominating it was created. Thus, in both cases, a new label~$l$ is added to~$L(v,m+1)$ in iteration~$m$, and we let~$P$ denote the corresponding $s$-$v$-path of length at most~$m+1$ obtained using the predecessor arcs and label numbers as in the proof of Proposition~\ref{prop:iso2}. Then, that~$P$ must actually have length exactly~$m+1$ since otherwise, the label~$l$ would have already been added in an earlier iteration. Since $|R|=m$, this implies that at least one arc~$r\in R$ is used more than once in~$P$. Thus, the $\alpha(r)$-$\alpha(r)$-subpath starting with the first traversal of~$r$ and ending right before the second traversal of~$r$ is a zero-duration cycle~$C$ reachable from~$s$. 
    
    Let~$P=(P_1,C,P_2)$ and let~$P'=(P_1,P_2)$ denote the $s$-$v$-path of length at most~$m$ obtained from~$P$ by removing the cycle~$C$.
    Thus, a label corresponding to an image that weakly dominates the image of~$P'$ is contained in~$L(v, m)$, which implies that either $f_j(P) < f_j(P')$ must hold for at least one objective~$j$ with $\dir_j = \min$, or $f_j(P) > f_j(P')$ must hold for at least one objective~$j$ with $\dir_j = \max$, since otherwise, the label~$l$ obtained from~$P$ would not be added to~$L(v,m+1)$ in iteration~$m$. 
    By commutativity of addition, $f_j(P) = f_j(P') + f_j(C)$, hence, if $\dir_j = \min$, it follows that $f_j(C)<0$, while if $\dir_j = \max$, then $f_j(C)>0$.
    In both cases, we obtain that~$C$ is an improving cycle in objective~$j$, which finishes the proof.
\end{proof}

\begin{restatable}{proposition}{propstrictlyisoaugcycleequiv}\label{prop:5}
If the temporal graph~$G$ contains no improving cycle, then Algorithm~\ref{list:alg2} terminates in line~18 and, for each $v\in V$, the returned label set~$L(v,k)$ contains a corresponding label for each image of an efficient $s$-$v$-path and only such labels.
\end{restatable}
\begin{proof}
    Since Proposition~\ref{prop:3} and Lemma~\ref{lemma:4} show that Algorithm~\ref{list:alg2} terminates in line~20 if and only if the temporal graph~$G$ contains an improving cycle. If~$G$ contains no improving cycle, the algorithm must therefore terminate in line~18, in which case we know that, for each $v\in V$, the returned label set~$L(v,k)$ contains a corresponding label for each image of an efficient $s$-$v$-path and only such labels.
\end{proof}

\noindent
Propositions~\ref{prop:3} and~\ref{prop:5} immediately yield the main theorem of this subsection.
\begin{theorem}\label{theorem:strict_iso_main_th}
    If the temporal graph~$G$ contains an improving cycle, 
    then Algorithm~\ref{list:alg2} correctly reports its existence. Otherwise, for each node~$v \in V$, the set~$L(v,k)$ returned by Algorithm~\ref{list:alg2} contains a corresponding label for each image of an efficient $s$-$v$-path and only such labels.
\end{theorem}

\section{Conclusion}
\label{sec:conclusion}

In this paper, we present the first algorithms for the single-source multiobjective temporal shortest path problem that do not rely on monotonicity or isotonicity assumptions on the objectives. Our results show that removing these structural assumptions fundamentally changes the behavior of the problem. This already holds for static graphs, where our algorithms can also be applied. Even for isotonic but non-monotone objectives, phenomena arise that have no analogue in static shortest path problems with additive objectives. In particular, arbitrarily long paths may need to be considered to generate all nondominated images, even in the absence of improving cycles. As a consequence, it remains open whether a label correcting algorithm can guarantee termination within a predefined number of iterations in general while still identifying all nondominated images or detecting improving cycles.

An interesting question for future research therefore concerns the computational complexity of deciding whether a given temporal graph with general objectives contains an improving cycle. Another natural direction is to identify additional sufficient conditions under which no bound on the maximum admissible path length is required. Finally, while our label correcting algorithms extend naturally to the all-pairs variant of the multiobjective temporal shortest path problem, designing more efficient algorithms for this setting remains an interesting topic 
for future work.

\bibliography{bibfile}

\end{document}